\documentclass[a4paper,twocolumn]{quantumarticle}
\pdfoutput=1

\usepackage[english]{babel}
\usepackage{amsmath,amssymb,amsthm,amscdx,graphicx,hyperref}

\DeclareMathOperator*{\argmin}{argmin}
\DeclareMathOperator{\Tr}{Tr}
\DeclareMathOperator{\ddi}{\mathtt{ddi}}
\DeclareMathOperator{\diff}{d}

\newtheorem{dfn}{Definition}
\newtheorem{lmm}{Lemma}
\newtheorem{thm}{Theorem}

\begin{document}

\title{On the role of designs in the data-driven approach to
  quantum statistical inference}

\author{Michele Dall'Arno}

\email{michele.dallarno.mv@tut.jp}

\affiliation{Department of Computer Science and Engineering,
  Toyohashi University of Technology, Toyohashi, Japan}

\affiliation{Yukawa Institute for Theoretical Physics, Kyoto
  University, Kyoto, Japan}

\begin{abstract}
  Designs,  and  in  particular  symmetric,  informationally
  complete (SIC)  structures, play an important  role in the
  quantum   tomographic  reconstruction   process  and,   by
  extension,  in certain  interpretations of  quantum theory
  focusing  on such  a process.   This  fact is  due to  the
  symmetry of  the reconstruction formula that  designs lead
  to.  However, it  is also known that  the same tomographic
  task,  albeit  with  a  less  symmetric  formula,  can  be
  accomplished   by   any  informationally   complete   (non
  necessarily symmetric)  structure.  Here we show  that, if
  the  tomographic   task  is  replaced  by   a  data-driven
  inferential approach,  the reconstruction,  while possible
  with  designs,  cannot  by   accomplished  anymore  by  an
  arbitrary informationally  complete structure.   Hence, we
  propose the  data-driven inference  as the arena  in which
  the role  of designs  naturally emerges.   Our inferential
  approach is  based on a minimality  principle according to
  which, among  all the possible inferences  consistent with
  the data, the weakest should be preferred, in the sense of
  majorization theory and statistical comparison.
\end{abstract}

\maketitle

\section{Introduction}

We consider  the scenario  in which  a correlation  (i.e.  a
conditional probability  distribution) between an  input and
an output is given. We  regard this correlation as generated
by some (unspecified) quantum  measurement upon the input of
some  (also unspecified)  states.  In  this context,  we say
that  a  measurement  is \emph{consistent}  with  the  given
correlation if  there exist states  upon the input  of which
the   measurement  produces   the  correlation.    Generally
speaking,  our  aim  is  to   produce  an  inference  for  a
consistent  measurement  (of  course,  we  could  adopt  the
opposite approach  of inferring  the states).   However, the
measurement  consistent with  any given  correlation is,  of
course, in general not unique.  How should we proceed?

For   instance,   the   protocol  of   quantum   measurement
tomography~\cite{LS99,  BCDFP09}  addresses  this  issue  by
additionally imposing  that the  given correlation  has been
generated by  a given set  of states.  The linearity  of the
theory allows  then in principle to  recover the measurement
by linear  inversion.  Of  course, in this  case the  set of
states  cannot   itself  be   obtained  via   quantum  state
tomography,  because the  latter,  in  a symmetric  fashion,
would  require an  assumption on  the measurement  which, by
definition of  our problem,  is instead unspecified  and the
target of the inference.   In other words, tomography cannot
``bootstrap'' itself.

In this work  we adopt a different  approach, reminiscent of
Jaynes'  maximum  entropy  principle~\cite{Jay57a,  Jay57b}.
Among   all   measurements   consistent   with   the   given
correlation, we  infer the \emph{minimally  committal} ones,
that  is,  informally, those  that  are  consistent with  as
little else  as possible  other than the  given correlation.
We  formalize  the  committal  degree of  a  measurement  by
regarding each measurement as a map~\cite{Oza80} assigning a
probability  distribution (on  the measurement  outcomes) to
any  state.   The  probability  range of  a  measurement  is
therefore  the  set of  all  correlations  a measurement  is
consistent  with,   and  our   goal  becomes  to   find  the
measurements whose  range, while  still including  the given
correlations, is minimal in  volume.  Since this inferential
protocol does  not require  any additional input  other than
the   given    correlation,   it    is   referred    to   as
data-driven~\cite{BD19}, and it can  be used for instance to
bootstrap the  tomographic protocol  in the  sense discussed
above.

The choice of comparing measurements based on their range is
based on the  operational role of the  range in majorization
theory      and     statistical      comparison~\cite{MOA79,
  Tor91}. Indeed, for any  two measurements, range inclusion
is a necessary and  sufficient condition for one measurement
to  be  able  to  simulate  the  other  through  a  suitable
statistical transformation~\cite{BKDPW05, Bla53, Bus12}.  In
this sense, our  goal can be reframed as  finding, among all
the measurements consistent with the given correlations, the
most universally simulable.

This optimization  problem has already been  solved, for any
given    correlation,     for    the    case     of    qubit
measurements~\cite{DBBT18,  DHBS19}.  Here,  we address  the
arbitrary  dimensional case.   Our main  result consists  of
showing  that  any  measurement   that  produces  the  given
correlation upon the input of  a spherical $2$-design set of
states  is  minimally   committal.   Technically,  the  main
ingredient  in   our  proof   is  a  result   within  John's
theory~\cite{Joh48} on minimum  volume enclosing ellipsoids.
In this sense, the  data-driven inference protocol, although
by definition a search  for the least committal measurement,
turns out  to be equivalent  to the linear inversion  at the
hearth of  measurement tomography under the  hypothesis that
the set of states forms a design.

At the foundational  level, our results shines  new light on
the role played by symmetric, informationally complete (SIC)
structures,  and more  generally  designs and  the class  of
morphophoric    measurements~\cite{SS20,   SS23}    recently
introduced by S\/lomczy\'nski and  Szymusiak, in the quantum
Bayesian interpretation (QBism)~\cite{Fuc10, FS13, CMS14} of
quantum  theory.  So-far,  such  a role  has been  justified
based  on the  symmetry  of  the tomographic  reconstruction
formula  (inherited  by  the   symmetry  of  the  structures
themselves) when such structures  are adopted.  However, any
informationally   complete   (not   necessarily   symmetric)
structure  is  universal   for  tomographic  reconstruction,
albeit  with  a  less-symmetric  formula.   Instead,  if  the
tomographic   task   is   replaced  with   the   data-driven
inferential task we consider, as a consequence of our result
not any informationally complete  structure will do, and thus
the role of designs emerges naturally.

\section{Formalization}

In order to  introduce the data-driven inference  map, it is
convenient  to formulate  operational concepts  from quantum
theory,  such as  states  and  measurements, in  geometrical
terms.    Table~\ref{tab:conversion}  provides   conversions
between  the  Hilbert-space  formalism and  the  geometrical
formalism of quantum theory.
\begin{table}[h!]
  \begin{center}
  \begin{tabular}{|c|c|c|}
    \hline
    & Hilbert form. & Geom.\\ 
    \hline
    Lin. space & Hermitian $d \times d$ & $\mathbb{R}^{\ell = d^2}$\\
    \hline
    Unit effect & $\openone_d$ & $\mathbf{u}_{\ell}$\\
    \hline
    Normalization & $\Tr[\rho] = 1$ & $\mathbf{u}_{\ell} \cdot \mathbf{s} = 1$\\
    \hline
    Measur. & $\sum_{j=1}^n \pi_j = \openone_d$ & $M^T \mathbf{u}_n = \mathbf{u}_{\ell}$\\
    \hline    
    Inner prod. & Hilbert-Schmidt & Dot prod.\\
    \hline
    Born rule & $p_j = \Tr[ \rho \pi_j ]$ & $\mathbf{p} = M \mathbf{s}$\\
    \hline
    Purity & $\Tr [ \rho^2 ]$ & $|\mathbf{s}|^2$\\
    \hline
  \end{tabular}
  \end{center}
  \caption{Conversion   table   between  the   Hilbert-space
    formalism  and  the  geometrical  formalism  of  quantum
    theory.}
  \label{tab:conversion}
\end{table}
Figure~\ref{fig:geometry} provides  a quick overview  of the
geometrical formalism of quantum theory.
\begin{figure}[h!]
  \includegraphics[width=\columnwidth]{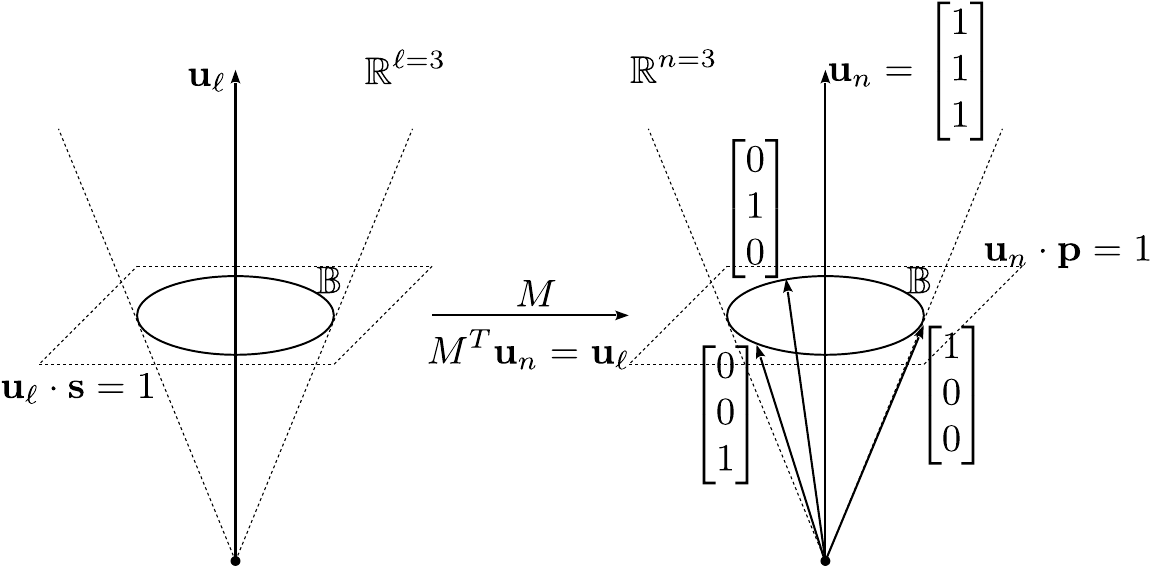}
  \caption{Geometrical  formalism of  quantum theory.   Left
    side: the linear space  of states and effects, including
    the unit effect $\mathbf{u}_{\ell}$, the hyperplane $\mathbf{u}_{\ell}
    \cdot \mathbf{s}  = 1$  where states  lie, and  the ball
    $\mathbb{B}$ on  whose surface  pure states  lie.  Right
    side:  the  probability   space,  including  the  vector
    $\mathbf{u}_n$   with   unit  elements,   the   hyperplane
    $\mathbf{u}_n   \cdot    \mathbf{p}$   where   probability
    distributions   lie,   and  the   extremal   probability
    distributions.  Any  measurement $M$  is a map  from the
    set $\mathbb{S}$ of admissible states to the probability
    space.}
  \label{fig:geometry}
\end{figure}
States   can  be   represented   by  real   vectors  in   an
$\ell$-dimensional real  space.  An  $n$-outcome quasi-measurement
$M$ is an $n \times \ell$ real matrix that associates to any
state $\mathbf{s}$ the quasi-probability distribution
\begin{align}
  \label{eq:born}
  \mathbf{p} := M \mathbf{s},
\end{align}
that is,  the entries of  $\mathbf{p}$ sum up to  the unity,
but  are not  necessarily positive.   The relaxation  of the
positivity constraint is typical of inferential protocols --
for example,  it is  shared by the  linear inversion  at the
heart of quantum tomography --, and is usually remedied by a
successive  search  of   the  closest  positivity-preserving
measurement,   according   to  some   relevant   operational
criterion. A quasi-measurement $M$  is said to be consistent
with quasi-probability distribution $\mathbf{p}$ if and only
if   there   exists   state  $\mathbf{s}$   that   generates
$\mathbf{p}$    when    measured     by    $M$,    as    per
Eq.~\eqref{eq:born}.  The  existence of a  unit measurement,
say $\mathbf{u}_{\ell}$, such that
\begin{align}
  \label{eq:states}
  \mathbf{u}_{\ell} \cdot \mathbf{s} = 1,
\end{align}
for   any  state   $\mathbf{s}$,   implies   that  the   set
$\mathbb{S}$ of  admissible states  lies on  the hyper-plane
orthogonal to  $\mathbf{u}_{\ell}$. We will  assume, without
loss of generality,  that $\mathbb{S}$ is a  spanning set of
$\mathbb{R}^\ell$.   Upon  choosing  $\mathbf{u}_{\ell}  \in
\mathbb{R}^{\ell}$ and $\mathbf{u}_n \in \mathbb{R}^n$ to be
the vectors whose elements are  all ones, the hyper-plane of
states    and   the    hyper-plane   of    quasi-probability
distributions  coincide,  which  will allow  for  a  unified
discussion.  The  fact that  for any state  $\mathbf{s}$ one
has $\mathbf{u}_n^T M  \mathbf{s} = 1$ is  equivalent to the
condition
\begin{align}
  \label{eq:identity}
  M^T \mathbf{u}_n = M^+ M \mathbf{u}_{\ell},
\end{align}
where $M^+$  denotes the Moore-Penrose  pseudoinverse.  From
Eq.~\eqref{eq:born},  informational   completeness  (IC)  is
equivalent  to the  condition that  $M$ is  left invertible,
that is, $M^+ M  = \openone_{\ell}$, where $\openone_{\ell}$
denotes the $\ell \times  \ell$ identity operator, and hence
$n \ge \ell$. For any  IC quasi-measurement $M$ one has that
$(\det M^T M)^{1/2}$  denotes, up to a  constant factor, the
volume of its probability range.

For     any    spanning     set    $\mathbb{S}     \subseteq
\mathbb{R}^{\ell}$   of  admissible   states  and   any  set
$\mathcal{P}   \subseteq    \mathbb{R}^n$   of   probability
distributions  spanning an  $\ell$-dimensional subspace,  we
denote with  $\mathcal{M}_{\mathbb{S}} ( \mathcal{P}  )$ the
set of quasi-measurements  from $\mathbb{S}$ consistent with
$\mathcal{P}$, that is
\begin{align*}
 \mathcal{M}_{\mathbb{S}}  \left(   \mathcal{P}  \right)  :=
 \left\{  M   \in  \mathbb{R}^{n  \times  \ell}   \Big|  M^T
 \mathbf{u}_n   =   \mathbf{u}_{\ell}   \wedge   \mathcal{P}
 \subseteq M \mathbb{S} \right\}.
\end{align*}

We  are  now in  a  position  to introduce  the  data-driven
inference map that, upon the input of any set of probability
distributions,  outputs the  set of  quasi-measurements that
are consistent with the input and minimally committal.

\begin{dfn}[Data-driven inference]
  Upon   the    input   of   any   set    $\mathcal{P}$   of
  quasi-probability          distributions          spanning
  $\mathbb{R}^{\ell}$,   the  output   of  the   data-driven
  inference map $\ddi_\mathbb{S} ( \mathcal{P} )$ is the set
  of quasi-measurements  consistent with  $\mathcal{P}$ with
  minimum-volume probability range, that is
  \begin{align}
    \label{eq:ddi}
    \ddi_\mathbb{S} \left( \mathcal{P} \right) := \argmin_{M
      \in   \mathcal{M}_{\mathbb{S}}    \left(   \mathcal{P}
      \right)} \det M^T M.
  \end{align}
\end{dfn}

\section{Main result}

In order to introduce our  main result, we need to formulate
some additional  operational concept from quantum  theory in
geometrical  terms.   The  purity of  a  state  $\mathbf{s}$
coincides with  its squared  norm $| \mathbf{s}  |^2$.  Pure
states, that is states with  unit purity $| \mathbf{s} |^2 =
1$, lie  on the  surface of  the ball  $\mathbb{B}$ obtained
intersecting the linear  constraint in Eq.~\eqref{eq:states}
with the cone, whose axis is $\mathbf{u}_{\ell}$, given by
\begin{align}
  \label{eq:cone}
  f \left( \mathbf{v} \right) := \left| \mathbf{v} \right|^2
  - \left( \mathbf{u}_{\ell} \cdot \mathbf{v} \right)^2 \le 0.
\end{align}
Notice  that  such  a  ball $\mathbb{B}$  is  in  general  a
superset of  the set  $\mathbb{S}$ of admissible  states.  A
probability  distribution $p$  over a  set $\mathcal{S}$  of
states  is a  spherical  $2$-design  if and  only  if it  is
indistinguishable  from  the  uniform  distribution  on  the
surface of $\mathbb{B}$ when given two copies. In formula
\begin{align}
  \sum_{\mathbf{s}  \in  \mathcal{S}}  p  \left(  \mathbf{s}
  \right)  \mathbf{s} \otimes  \mathbf{s}  =  \int O  \left(
  \mathbf{s}   \otimes   \mathbf{s}    \right)   O^T   \diff
  O\label{eq:design},
\end{align}
where $\mathbf{s}$ in the right-hand side is any pure state,
$O$    denotes    any    orthogonal   matrix    such    that
$\hat{\mathbf{u}}_{\ell}^T  O  \hat{\mathbf{u}}_{\ell} =  1$
(we  write $\hat{\cdot}$  to  denote the  unit vector),  and
$\diff O$  denotes the  invariant measure on  the orthogonal
group  $O(\ell -  1)$ in  the subspace  of $\mathbb{R}^\ell$
satisfying Eq.~\eqref{eq:states}.

We are  now in  a position  to state  our main  result, that
shows that the quasi-measurements  that output the given set
of probability  distributions upon the input  of a spherical
$2$-design are minimally committal.

\begin{thm}
  \label{thm:ddi}
  Upon   the    input   of   any   set    $\mathcal{P}$   of
  quasi-probability          distributions          spanning
  $\mathbb{R}^{\ell}$, quasi-measurement $M$  belongs to the
  output of the data driven inference map $\ddi_\mathbb{S} (
  \mathcal{P}  )$  if  the  counter-image  $\mathcal{S}$  of
  $\mathcal{P}$, that is
  \begin{align}
    \label{eq:master}
    \mathcal{P} = M \mathcal{S},
  \end{align}
  supports a spherical $2$-design.
\end{thm}

Notice that  Eq.~\eqref{eq:master} represents  a closed-form
characterization  of  $M$  whenever  $\mathcal{P}$  contains
$\ell$   linearly  independent   probability  distributions.
Indeed, in this case the  only $\mathcal{S}$ that supports a
spherical   $2$-design   is   the   regular   simplex,   and
Eq.~\eqref{eq:master} can  be inverted to  explicitly obtain
$M$.

\section{Proof of main result}

The following commuting diagram  summarizes the statement of
Thm.~\ref{thm:ddi}  (left vertical  arrow), as  well as  the
statements  of  the  three   Lemmas  (horizontal  and  right
vertical arrows) in which we split its proof.
\begin{align*}
  \CDfattrue
  \begin{CD}
    \mathcal{S}     \textrm{     sph.     $2$-design}     @Z
    \textrm{Lemma~\ref{lmm:design}}  ZZ O  \left(  \ell -  1
    \right)  \subseteq \ddi_{\mathbb{B}}  \left( \mathcal{S}
    \right)   \\   @V  \textrm{Thm.~\ref{thm:ddi}}   VV   @V
    \textrm{Lemma~\ref{lmm:outer}}      VV\\      M      \in
    \ddi_{\mathbb{S}}   \left(    \mathcal{P}   \right)   @Z
    \textrm{Lemma~\ref{lmm:commuting}}   ZZ  \openone_{\ell}
    \in \ddi_{\mathbb{S}} \left( \mathcal{S} \right)
  \end{CD}.
\end{align*}
Notice that, on the right hand side of the diagram (that is,
in the three  lemmas), the $\ddi$ map is applied  to the set
$\mathcal{S}$ of states, rather than to a set of probability
distributions.   This  is  consistent  with  our  choice  of
representing  states as  probability  distributions, as  per
Eq.~\eqref{eq:states}.

First, we modify a  proof technique used~\cite{Joh48} in the
related context  of minimum  volume enclosing  ellipsoids to
show  that  spherical  $2$-designs   lead  to  a  sufficient
condition for the  inclusion of the orthogonal  group in the
output  of the  data-driven inference  on the  ball of  pure
states.  Measurements  corresponding to  orthogonal matrices
are symmetric, informationally complete (SIC).

\begin{lmm}
  \label{lmm:design}
  Upon the  input of  any set  $\mathcal{S}$ of  states that
  supports a spherical design, the output of the data-driven
  inference map  $\ddi_{\mathbb{B}} ( \mathcal{S} )$  on the
  ball $\mathbb{B}$ includes any  orthogonal matrix $O$ such
  that $\hat{\mathbf{u}}_{\ell}^T  O \hat{\mathbf{u}}_{\ell}
  = 1$.
\end{lmm}

\begin{proof}
  Since  $\mathcal{S}$  supports   a  spherical  $2$-design,
  $\mathcal{S}$  is a  spanning set  of $\mathbb{R}^{\ell}$,
  and  hence any  $M$ in  the domain  of optimization  is an
  $\ell$-outcomes  invertible  quasi-measurement.   For  any
  $\ell$-outcome   invertible   quasi-measurement  $M$   the
  condition   $\mathcal{S}   \subseteq  M   \mathbb{B}$   is
  equivalent to the  condition $M^{-1} \mathcal{S} \subseteq
  \mathbb{B}$, which in turn is equivalent to the conditions
  $f  ( M^{-1}  \mathbf{s} )  \le 0$  and $\mathbf{u}_{\ell}
  \cdot   \mathbf{s}   =   1$  for   any   $\mathbf{s}   \in
  \mathcal{S}$.   From  Eq.~\eqref{eq:cone} one  immediately
  has
  \begin{align*}
    0  \ge   f  \left(  \mathbf{s}  \right)   =  \Tr  \left[
      \mathbf{s}     \otimes     \mathbf{s}    \right]     -
    \mathbf{u}_{\ell}^T \left( \mathbf{s} \otimes \mathbf{s}
    \right) \mathbf{u}_{\ell}.
  \end{align*}
  Hence
  \begin{align*}
    0 \ge  & f \left( M^{-1}  \mathbf{s} \right) \\ =  & \Tr
    \left[  M^{-1}  \left(   \mathbf{s}  \otimes  \mathbf{s}
      \right)  M^{-T} \right]  - \mathbf{u}_{\ell}^T  M^{-1}
    \left(  \mathbf{s}  \otimes  \mathbf{s}  \right)  M^{-T}
    \mathbf{u}_{\ell}.
  \end{align*}
  Due  to Lemma~\ref{lmm:average}  there exists  probability
  distribution $p$ such that
  \begin{align*}
    0  \ge  &  \sum_{\mathbf{s} \in  \mathcal{S}}  p  \left(
    \mathbf{s}  \right) f  \left( M^{-1}  \mathbf{s} \right)
    \nonumber\\ =  & \frac1{\ell}  \left( \Tr  \left[ M^{-1}
      M^{-T}  \right]  - \mathbf{u}_{\ell}^T  M^{-1}  M^{-T}
    \mathbf{u}_{\ell} \right).
  \end{align*}
  Due  to  Lemma~\ref{lmm:inversion}   also  $M^{-1}$  is  a
  quasi-measurement,   hence  $M^{-T}   \mathbf{u}_{\ell}  =
  \mathbf{u}_{\ell}$.  Since  $|\mathbf{u}_{\ell}|^2 = \ell$
  one has
  \begin{align}
    \label{eq:proof2}
    0  \ge \Tr  \left[ M^{-1}  M^{-T} \right]  - \ell  = \Tr
    \left[ M^{-1} M^{-T} - \openone_{\ell} \right].
  \end{align}
  
  Notice  that,  for   any  $X  >  0$,  one   has  $\Tr[X  -
    \openone_{\ell}] \ge \log \det  X$, with equality if and
  only  if $X  =  \openone_{\ell}$.  Using  this fact,  from
  Eq.~\eqref{eq:proof2} one has the following majorization
  \begin{align*}
    0 \ge \ln \det M^{-1} M^{-T},
  \end{align*}
  or equivalently
  \begin{align}
    \label{eq:proof3}
    \det M^T M \ge 1.
  \end{align}
  
  Summarizing,  Eq.~\eqref{eq:proof3}  shows that,  for  any
  $\ell$-outcome  quasi-measurement  $M$  in the  domain  of
  optimization of $\ddi_{\mathbb{B}}  ( \mathcal{S} )$, that
  is   $\mathcal{S}  \subseteq   M  \mathbb{B}$,   the  cost
  function, that  is $\det M^T  M $, attains a  value larger
  than or  equal to one, that  is the value attained  by the
  orthogonal group, thus proving the statement.
\end{proof}

Notice that  also the reverse of  Lemma~\ref{lmm:design} can
be proved by  similarly modifying a proof  technique used in
Ref.~\cite{Joh48}; however,  since such  a statement  is not
necessary in  order to prove Theorem~\ref{thm:ddi},  we omit
such a proof here.

Next,  we show  that an  outer approximation  of the  set of
admissible states in terms of  the ball of pure states leads
to   a   sufficient   condition   for   the   inclusion   of
quasi-measurement $\openone_{\ell}$ (the  $\ell \times \ell$
identity  operator)   in  the  output  of   the  data-driven
inference map.

\begin{lmm}
  \label{lmm:outer}
  Upon the input  of any given set  $\mathcal{S}$ of states,
  if   the  output   of   the   data-driven  inference   map
  $\ddi_{\mathbb{B}}   (   \mathcal{S}   )$  on   the   ball
  $\mathbb{B}$  includes the  quasi-measurement {\normalfont
    $\openone_{\ell}$},   then  the   the   output  of   the
  data-driven inference map $\ddi_{\mathbb{S}} ( \mathcal{S}
  )$  on  the set  $\mathbb{S}$  of  admissible states  also
  includes      the      quasi-measurement      {\normalfont
    $\openone_{\ell}$}.
\end{lmm}

\begin{proof}
  To  prove  the  statement,   we  proceed  by  reductio  ad
  absurdum.

  Since  by   negation  of  the  thesis   $\openone_{\ell}  \not\in
  \ddi_{\mathbb{S}}      (\mathcal{S})$,      either      i)
  $\ddi_{\mathbb{S}}  (\mathcal{S})  =  \emptyset$,  or  ii)
  there exists a  quasi-measurement $M \in \ddi_{\mathbb{S}}
  (\mathcal{S})$ such that the value of the cost function in
  $M$ is strictly smaller than the value in $\openone_{\ell}$, that
  is, $\det M^T M < \det \openone_{\ell} = 1$.  Since $\mathcal{S}$
  is  a   subset  of  $\mathbb{S}$,   the  quasi-measurement
  $\openone_{\ell}$  is a  feasible  point of  the optimization  of
  $\ddi_{\mathbb{S}}  (   \mathcal{S}  )$,   thus  excluding
  alternative i) and leaving us with alternative ii).
    
  Since  any  set $\mathbb{S}$  of  admissible  states is  a
  subset  of the  ball  $\mathbb{B}$ on  whose surface  pure
  states  lie, that  is  $\mathbb{S} \subseteq  \mathbb{B}$,
  from  the consistency  condition $\mathcal{S}  \subseteq M
  \mathbb{S}$ one has  $\mathcal{S} \subseteq M \mathbb{B}$.
  Hence, quasi-measurement  $M$ is  a feasible point  of the
  optimization  of  $\ddi_{\mathbb{B}}   (  \mathcal{S}  )$.
  Since $\det  M^T M  < \det \openone_{\ell}  = 1$,  one has
  that     $\openone_{\ell}    \not\in     \ddi_{\mathbb{B}}
  (\mathcal{S})$, that contradicts the hypothesis.
\end{proof}

Notice   that   reversing   the   logical   implication   of
Lemma~\ref{lmm:outer}  the   statement  fails   in  general;
therefore,  Lemma~\ref{lmm:outer}  is  the  reason  why  the
logical implication of Theorem~\ref{thm:ddi} is one-way only
for arbitrary  dimension.  However,  for qubits it  is clear
that the logical implication of Lemma~\ref{lmm:outer} can be
reversed,   leading    in   this   case   to    a   complete
characterization  of  minimally  committal  measurements  in
terms of spherical $2$-designs.

Third,   we    recast   the    inclusion   of    any   given
quasi-measurement in the output of the data-driven inference
map as the inclusion of measurement $M = \openone_{\ell}$ in
the output of the data-driven inference map.

\begin{lmm}
  \label{lmm:commuting}
  Upon  the input  of any  set $\mathcal{P}$  of probability
  distributions   spanning   an  $\ell$-dimensional   space,
  quasi-measurement  $M$ belongs  to the  output of  the map
  $\ddi_{\mathbb{S}} (  \mathcal{P} )$  if and only  if upon
  the   input   of   the  set   $\mathcal{S}$   of   states,
  quasi-measurement {\normalfont  $\openone_{\ell}$} belongs
  to the output of  the map $\ddi_{\mathbb{S}} ( \mathcal{S}
  )$, where $\mathcal{S}$ is given by Eq.~\eqref{eq:master}.
\end{lmm}

\begin{proof}
  The proof  is split in two  parts.  In the first  part, we
  prove  that quasi-measurement  $M$ is  bijective from  the
  optimization   domain   $\mathcal{M}_{\mathbb{S}}  (   M^+
  \mathcal{P} )$ of map $\ddi_{\mathbb{S}} ( M^+ \mathcal{P}
  )$ to the  optimization domain $\mathcal{M}_{\mathbb{S}} (
  \mathcal{P} )$ of map $\ddi_{\mathbb{S}} ( \mathcal{P} )$.
  In the  second part,  we prove that  quasi-measurement $M$
  preserves the ordering induced  by the cost function given
  by the range volume.

  For any quasi-measurement  $L \in \mathcal{M}_{\mathbb{S}}
  ( M^+ \mathcal{P} ) $ one has that:
  \begin{itemize}
  \item $M L$ is  informationally complete, that is, $(ML)^+
    (ML) = \openone_{\ell}$, as  it immediately follows from
    the informational completeness of $M$ and $L$.
  \item  $M  L$ is  a  quasi-measurement,  that is,  $(ML)^T
    \mathbf{u}_n = \mathbf{u}_{\ell}$. Indeed
    \begin{align*}
      \left(   M  L   \right)^T  \mathbf{u}_n   =  L^T   M^T
      \mathbf{u}_n     =     L^T     \mathbf{u}_{\ell}     =
      \mathbf{u}_{\ell},
    \end{align*}
    where  the second  equality follows  from the  fact that
    $M^T  \mathbf{u}_n =  \mathbf{u}_{\ell}$  and the  third
    equality    follows   from    the    fact   that    $L^T
    \mathbf{u}_{\ell} = \mathbf{u}_{\ell}$.
  \item  $M L$  is  consistent with  $\mathcal{P}$, that  is
    $\mathcal{P}  \subseteq M  L \mathbb{S}$.   Indeed, from
    the hypothesis $M^+  \mathcal{P} \subseteq L \mathbb{S}$
    one has  $M M^+  \mathcal{P} \subseteq M  L \mathbb{S}$,
    and from the fact that $M M^+ \mathcal{P} = \mathcal{P}$
    one has $\mathcal{P} \subseteq M L \mathbb{S}$.
  \end{itemize}
  Hence
  \begin{align}
    \label{eq:inclusion1}
    M   \mathcal{M}_{\mathbb{S}}   \left(  M^+   \mathcal{P}
    \right)   \subseteq    \mathcal{M}_{\mathbb{S}}   \left(
    \mathcal{P} \right).
  \end{align}
  Since     $M^+     M      =     \openone_{\ell}$,     from
  Eq.~\eqref{eq:inclusion1}        one       also        has
  \begin{align}
    \label{eq:inclusion2}
    \mathcal{M}_{\mathbb{S}} \left(  M^+ \mathcal{P} \right)
    \subseteq     M^+    \mathcal{M}_{\mathbb{S}}     \left(
    \mathcal{P} \right).
  \end{align}

  For any quasi-measurement  $N \in \mathcal{M}_{\mathbb{S}}
  ( \mathcal{P} ) $ one has that:
  \begin{itemize}
    \item  $M^+ N$  is  informationally  complete, that  is,
      $(M^+ N)^+ (M^+ N) = \openone_{\ell}$.  Indeed
      \begin{align*}
        \left( M^+ N \right)^+ \left( M^+  N \right) = N^+ M M^+
        N = N^+ N = \openone_{\ell},
      \end{align*}
      where the second equality  follows from the facts that
      $M  M^+   \mathcal{P}  =   \mathcal{P}$  and   $N  N^+
      \mathcal{P} = \mathcal{P}$.
    \item  $M^+ N$  is a  quasi-measurement, that  is, $(M^+
      N)^T \mathbf{u}_{\ell} = \mathbf{u}_{\ell}$. Indeed
      \begin{align*}
        \left( M^+ N \right)^T \mathbf{u}_{\ell} = N^T M M^+
        \mathbf{u}_n = N^T \mathbf{u}_n = \mathbf{u}_{\ell},
      \end{align*}
      where    the     first    equality     follows    from
      Lemma~\ref{lmm:inversion}, the second equality follows
      from  the fact  that  $M M^+  N =  N$,  and the  third
      equality follows from the  fact that $N^T \mathbf{u}_n
      = \mathbf{u}_{\ell}$.
    \item $M^+ N$ is consistent with $M^+ \mathcal{P}$, that
      is $M^+ \mathcal{P} \subseteq M^+ N \mathbb{S}$, as it
      immediately follows  from the  hypothesis $\mathcal{P}
      \subseteq N \mathbb{S}$.
  \end{itemize}
  Hence
  \begin{align}
    \label{eq:inclusion0}
    M^+ \mathcal{M}_{\mathbb{S}}  \left( \mathcal{P} \right)
    \subseteq     \mathcal{M}_{\mathbb{S}}    \left(     M^+
    \mathcal{P} \right).
  \end{align}
  Since $M^+  M N  = N$  for any  $N \in  \mathcal{N}$, from
  Eq.~\eqref{eq:inclusion0}   one   also  has
  \begin{align}
    \label{eq:inclusion3}
    \mathcal{M}_{\mathbb{S}}   \left(  \mathcal{P}   \right)
    \subseteq   M    \mathcal{M}_{\mathbb{S}}   \left(   M^+
    \mathcal{P} \right).
  \end{align}

  Combining                      Eqs.~\eqref{eq:inclusion0},
  \eqref{eq:inclusion1},              \eqref{eq:inclusion2},
  and~\eqref{eq:inclusion3}, one immediately has
  \begin{align*}
    \mathcal{M}_{\mathbb{S}} \left(  M^+ \mathcal{P} \right)
    &  =  M^+  \mathcal{M}_{\mathbb{S}}  \left(  \mathcal{P}
    \right),\\  \mathcal{M}_{\mathbb{S}} \left(  \mathcal{P}
    \right)  &  =   M  \mathcal{M}_{\mathbb{S}}  \left(  M^+
    \mathcal{P} \right),
  \end{align*}
  that   is,  quasi-measurement   $M$   is  bijective   from
  $\mathcal{M}_{\mathbb{S}}    (   M^+    \mathcal{P})$   to
  $\mathcal{M}_{\mathbb{S}} ( \mathcal{P})$, which concludes
  the first part of the proof.
  
  For any quasi-measurement  $L \in \mathcal{M}_{\mathbb{S}}
  ( M^+ \mathcal{P} ) $ by explicit computation one has that
  \begin{align*}
    \det \left( M L \right)^T \left(  M L \right) = \det M^T
    M \det L^T L.
  \end{align*}

  For any quasi-measurement  $N \in \mathcal{M}_{\mathbb{S}}
  ( \mathcal{P} ) $ one has that
  \begin{align*}
    & \det \left( M^+ N \right)^T \left( M^+ N \right)\\ = &
    \det N^T \left(M^+\right)^T M^+ N\\ = & \det O^T D^T V^T
    W^T E^T U^T  U E W V D  O\\ = & \det D^T  \left( V^T W^T
    E^T E W  V \right) D\\ = &  \det D^T D \det E^T  E\\ = &
    \left( \det M^T M \right)^{-1} \det N^T N,
  \end{align*}
  where we made use of  the singular value decompositions $N
  = V D  O$ and $M^+ =  U E W$ for  some orthogonal matrices
  $O, U \in \mathbb{R}^{\ell  \times \ell}$, some isometries
  $V, W^T \in \mathbb{R}^{n \times \ell}$, and some diagonal
  matrices $D, E \in  \mathbb{R}^{\ell \times \ell}$, and we
  have used  the fact  that, since $\mathcal{P}  \subseteq N
  \mathbb{S}$ and $\mathcal{P} =  M \mathcal{S}$, one has $N
  N^+  = M  M^+$ and  hence $V  V^+ =  W^+ W$,  thus $W  V =
  \openone_{\ell}$.

  Since  $\det  M^T  M$   is  a  positive  constant  factor,
  quasi-measurement  $M$ preserves  the ordering  induced by
  the cost function, which concludes the proof.
\end{proof}

This concludes  the proof  of Theorem~\ref{thm:ddi}.

\section{Conclusions and open problems}

In  this  work we  showed  that  a data-driven  approach  to
quantum statistical  inference is possible with  designs and
cannot by accomplished by arbitrary informationally complete
structures, thus suggesting that data-driven inference, more
than quantum tomography, should  be regarded as the scenario
in which  the role of designs  naturally emerges.  Similarly
to  Jaynes'  maximum   entropy  principle,  our  inferential
approach  is based  on a  minimality principle  according to
which, among all the possible inferences consistent with the
data,  the weakest  should  be preferred,  in  the sense  of
majorization theory and statistical comparison.

To conclude, we  believe it would be relevant  to extend our
analysis of the  role played by symmetric  structures in the
data-driven      inference      to     the      morphophoric
structures~\cite{SS20, SS23} recently introduced and studied
by   S\/lomczy\'nski  and   A.    Szymusiak.  Indeed,   such
structures  generalize designs  in a  way that  preserves at
least some of the properties analyzed in this work.

\section{Acknowledgments}

M.~D.   acknowledges support  from  Toyohashi University  of
Technology,  the  international   research  unit  in  Quantum
Information  in  Kyoto  University, the  MEXT  Quantum  Leap
Flagship Program  (MEXT Q-LEAP) Grant  No.  JPMXS0118067285,
and the JSPS KAKENHI Grant Number JP20K03774.

\appendix

\section{Properties of measurements and designs}

In  this section  we present  some elementary  properties of
quasi-measurements and spherical $2$-designs that we used in
the main text.

\begin{lmm}[Closure under inversion]
  \label{lmm:inversion}
  The   set   of    quasi-measurements   is   closed   under
  pseudoinversion,   that  is,   for  any   matrix  $M   \in
  \mathbb{R}^{n \times \ell}$ such  that $M^T \mathbf{u}_n =
  M^+  M  \mathbf{u}_{\ell}$,  one  also has  $(  M^+  )^{T}
  \mathbf{u}_{\ell} = M M^+ \mathbf{u}_n$.
\end{lmm}

\begin{proof}
  By hypothesis one has
  \begin{align*}
    M^T \mathbf{u}_n = M^+ M \mathbf{u}_{\ell}.
  \end{align*}
  By multiplying from the left by $( M^+ )^T$ one has
  \begin{align*}
    \left(  M^+  \right)^T  M^T \mathbf{u}_n  =  \left(  M^+
    \right)^T M^+ M \mathbf{u}_{\ell}.
  \end{align*}
  By using the identities $( M^+ )^T M^T = M M^+$ and $( M^+
  )^T M^+ M = ( M^+ )^T$ one has
  \begin{align*}
    \left(   M^+  \right)^T   \mathbf{u}_{\ell}   =  M   M^+
    \mathbf{u}_n,
  \end{align*}
  which proves the statement.
\end{proof}

\begin{lmm}[Spherical $2$-designs]
  \label{lmm:average}
  For any spherical $2$-design  $p$ over a set $\mathcal{S}$
  of states one has{\normalfont
  \begin{align*}
    \sum_{\mathbf{s}  \in \mathcal{S}}  p \left(  \mathbf{s}
    \right)     \mathbf{s}      \otimes     \mathbf{s}     =
    \frac{\openone_{\ell}}{\ell}.
  \end{align*}}
\end{lmm}

\begin{proof}
  From   the  invariance   under   transformations  in   the
  orthogonal group $O ( \ell -  1 )$, the right-hand side of
  Eq.~\eqref{eq:design} is given by
  \begin{align}
    \label{eq:design0}
    \int O \left( \mathbf{s}  \otimes \mathbf{s} \right) O^T
    \diff     O    =     \lambda    \left(     \openone_{\ell}    -
    \hat{\mathbf{u}}_{\ell}  \otimes \hat{\mathbf{u}}_{\ell}
    \right)    +    \nu   \hat{\mathbf{u}}_{\ell}    \otimes
    \hat{\mathbf{u}}_{\ell},
  \end{align}
  for some  constants $\lambda$  and $\nu$.   Constant $\nu$
  can   be   obtained   by   multiplying   both   sides   of
  Eq.~\eqref{eq:design0}  by  $\mathbf{u}_{\ell}^T$  to  the
  left  and  by  $\mathbf{u}_{\ell}$  to  the  right.   From
  Eq.~\eqref{eq:states},  the  left-hand  side  equals  one.
  Hence one has $1 = \nu | \mathbf{u}_{\ell} |^2$, and since
  $| \mathbf{u}_{\ell} |^2  = \ell$ one has  $\nu = 1/\ell$.
  Constant $\lambda$  can be obtained by  tracing both sides
  of Eq.~\eqref{eq:design0}.  Since for  pure states one has
  $| \mathbf{s}  |^2 =  1$, the  left-hand side  equals one.
  Hence  one has  $1 =  \lambda  ( \ell  -  1 )  + \nu$,  or
  equivalently $\lambda = \nu = 1/\ell$.
\end{proof}

Notice   that  the   statement  of   Lemma~\ref{lmm:average}
immediately follows from the fact (not assumed in the proof)
that regular  simplices are spherical $2$-designs.   In this
case, as exemplified by Fig.~\ref{fig:geometry}, it suffices
to consider the regular simplex of probability distributions
to immediately prove the statement of the lemma.

\end{document}